\documentclass{article}
\usepackage[utf8]{inputenc}
\usepackage{graphicx}
\usepackage{amsmath,amssymb}
\usepackage{arydshln}
\usepackage{booktabs}
\usepackage{verbatim}
\usepackage{enumitem}
\usepackage{bm}
\usepackage{wasysym}
\usepackage{esint}
\usepackage{xcolor}
\usepackage{tikz}
\usepackage{times} 
\usepackage{lmodern} 
\usepackage{hanging}
\usetikzlibrary{shapes,backgrounds,arrows,chains,matrix,positioning,scopes}
\usepackage{caption}
\usepackage{titlesec}
\usepackage{amsthm}
\newtheorem{theorem}{Theorem}[section]

\newtheorem{Definition}[theorem]{Definition}
\usepackage{geometry}
\date{}

\begin{document}
\begin{table}[ht]
\centering
\begin{tabular}{p{7 cm}p{6 cm}}
\multicolumn{2}{c}{\huge \textbf{Fairly Private Through Group Tagging }}                                                                      \\
\multicolumn{2}{c}{\huge \textbf{  and Relation Impact}} \\
&\\
&\\
&\\
\large  \textbf{Poushali Sengupta }                                 &\large  \textbf{Subhankar Mishra }                                     \\
&\\
 Department of Statistics                           & School of Computer Sciences                           \\
University of Kalyani                              & National Institute of Science, Education and Research \\
Kalyani, West Bengal                               & Bhubaneswar, India -752050                            \\
India - 741235                                     & Homi Bhabha National Institute,                       \\
 \multicolumn{1}{l}{}                               & Anushaktinagar, Mumbai - 400094, India                \\
 \multicolumn{1}{l}{tua.poushalisengupta@gmail.com} & smishra@niser.ac.in                                   \\
 \multicolumn{1}{l}{ORCID: 0000-0002-6974-5794}     & ORCID: 0000-0002-9910-7291                          
\end{tabular}
\end{table}

\begin{abstract}

Privacy and Fairness both are very important nowadays. For most of the cases in the online service providing system, users have to share their personal information with the organizations. In return, the clients not only demand a high privacy guarantee to their sensitive data but also expected to be treated fairly irrespective of their age, gender, religion, race, skin color, or other sensitive protected attributes. Our work introduces a novel architecture that is balanced among the privacy-utility-fairness trade-off. The proposed mechanism applies \textit{Group Tagging Method } and  \textit{Fairly Iterative Shuffling} \textit{(FIS)} that amplifies the privacy through random shuffling and prevents linkage attack. The algorithm introduces a fair classification problem by \textit{Relation Impact} based on \textit{Equalized Minimal FPR-FNR} among the protected tagged group. For the count report generation, the aggregator uses TF-IDF to add noise for providing longitudinal Differential Privacy guarantee. Lastly, the mechanism boosts the utility through risk minimization function and obtain the optimal privacy-utility budget of the system. In our work, we have done a case study on gender equality in the admission system and helps to obtain a satisfying result which implies that the proposed architecture achieves the group fairness and optimal privacy-utility trade-off for both the numerical and decision making Queries.

\end{abstract}
\section{Introduction}\label{intro}
To deal with information leakage \cite{b1} \cite{b2}, nowadays the organizations and industries have all concentration on the "privacy" of the dataset which is the most important and necessary thing to protect the dataset from some unexpected attacks by the attackers that causes data breach.
The main aim of a statistical disclosure control is to protect the "privacy" of "individuals" in a database from the sudden attack. To solve this problem, in 2006, Dwork introduces the notion of "Differential Privacy"(DP) \cite{b3}, \cite{b4} \cite{b5}. Some of the important works on DP have been summarized in table \ref{tab:dp}. 

\begin{table}[ht]
\begin{tabular}{p{0.26\textwidth}lp{0.6\textwidth}} \toprule
\textbf{Paper}                                                        & \textbf{Year} & \textbf{Brief Summary}                                                                                           \\ \midrule
Differential Privacy \cite{b5} & 2006          & Initial work on DP                                                                                               \\
RAPPOR \cite{b7}                         & 2014          & Fastest implementation of LDP                                                                                    \\
Deep Learning with Differential Privacy\cite{b8} & 2016 & Introduces a refined analysis of privacy costs within DP\\
Communication-Efficient Learning of Deep Networks from Decentralized Data \cite{b9} & 2016 &  Introduces practical method for the federated learning of deep networks based on iterative model averaging\\

PROCHLO \cite{b10}                     & 2017          & Encode Shuffle and Analyse (ESA) for strong CDP                                                                  \\
 Adding DP to Iterative Training Procedures \cite{b11} & 2018& Introduces an algorithm with the modular approach to minimize the changes on training algorithm, provides a variety of configuration strategies for the privacy mechanism.\\
Amplification by Shuffling \cite{b12}     & 2018          & $\mathcal{O}(\epsilon \sqrt{\log{\frac{1}{\delta}/ n}, \delta})$ centralized differential privacy(CDP) guarantee \\
Controlled Noise \cite{b13}       & 2019          &    duel algorithm based on average consensus (AC)                     \\ 
Capacity Bounded DP \cite{b14} & 2019 &  Location Privacy with matrix factorization\\
ARA \cite{b15} & 2020 & Aggregration of privacy reports \& fastest CDP.\\
Privacy at Scale\cite{b16} & 2020 & Privatization, ingestion and aggregation \\
ESA Revisited \cite{b35} & 2020 & Improved \cite{b12} with longitudinal privacy \\ 
BUDS\cite{b17} & 2020 & Optimize the privacy-utility trade-off by Iterative Shuffling and Risk Minimizer.\\
\hline                                                          \end{tabular}
\caption{Related works on Differential Privacy}
\label{tab:dp}
\end{table}

 
Adding to the importance of trade-off between privacy-Utility guarantee to the sensitive data, the demand for fairness also belongs to the clients' top priority list. Nowadays, in most online service systems, users not only ask for high privacy but also expect to be treated fairly irrespective of religion, skin color, native language, gender, cast, and other sensitive classes. Table \ref{tab:fair} shows some existing algorithms on fairness mechanism.

\begin{table}[ht]
\begin{tabular}{p{0.26\textwidth}lp{0.6\textwidth}} \toprule
\textbf{Paper}                                                        & \textbf{Year} & \textbf{Brief Summary}                                                                                           \\ \midrule

Fairness Beyond  Disparate  Treatment  Disparate  Impact \cite{b19}   & 2017         & Introduces Disparate Mistreatment for group fairness. \\

Algorithmic decision making and the cost of fairness \cite{b20}& 2017 &  Deals with a tension between constraint and unconstrained algorithms.\\
On  Fairness and  Calibration\cite{b21} & 2017 &  Explore the stress between minimizing error discrepancy across different population groups while maintaining calibrated probability estimates.\\

An  Information-Theoretic  Perspective Fairness and Accuracy \cite{b22}& 2019 & Provide trade-off between fairness and accuracy.\\
Fairness  Through the  Lens of  Proportional  Equality \cite{b23} & 2019 & Measure the fairness of a  classifier and provide group fairness through proportional equality \\
FACT: A Diagnostic for Group Fairness Trade-offs \cite{b24}& 2020 & Enable systematic characterization of performance-fairness trade-offs among the group. \\
Fair  Learning with  Private Demographic  Data \cite{b25} & 2020 & allows individuals to release their sensitive information privately while allowing the downstream entity to learn predictors which are non-discriminatory giving theoretical guarantees on fairness. \\
Fairness,  Equality,  and  Power in  Algorithmic Decision-Making \cite{b26} & 2021 &  Focuses on inequality and the causal impact of algorithms and the distribution of power between and within groups. \\
\hline         
\end{tabular}
\caption{Related works on Fairness}
\label{tab:fair}
\end{table}
 
In their 2012 work \cite{b18} Dwork et al. have explicitly shown that individual fairness is a generalization of $\epsilon-$differential privacy. There is a distinction in fair ML research between ‘group’ \& ‘individual’ fairness measures. Much existing research assumes that both have importance, but conflicting.
\cite{b34} argues that this seeming conflict is based on a
misconception and shows that individual and group fairness are not fundamentally in conflict.\cite{b28,b29,b30,b31,b32,b33} are some state-of-the-art algorithms on Fairness with Differential Privacy. However, none of them provide either the tight bounds on privacy \& loss function or an optimal trade-off between privacy-utility-fairness. Our proposed mechanism aims to achieve both.\\


In this work, the idea is to introduce a group fairness mechanism along with a differential privacy mechanism with the help of the Fairly Iterative Shuffling technique which provides optimal privacy-utility trade-off. The main aim of this work is to make the balance among privacy-utility-fairness during the process. The contributions of this work are:
\begin{enumerate}
    \item \textit{Group Tagging} method tags every individual by their protected attribute class so that the $FIS$ can not change their protected data during the shuffling. If it will change, the classifier's decision will not be fair as the system wants to generate group fairness among the protected attribute class in the decision-making problem.
    \item $FIS$ itself posses a strong  privacy guarantee to the data that prevents all types of linkage attack
    \item The proposed architecture applies the notion of \textit{Relation Impact} for maintaining group fairness to decide the allowance or giving the chance of participation of the individuals in any particular event without the leakage of their protective attributes.
    \item Here, every individual is treated according to their due rather than only "merit". So, the mechanism not only focuses on the directly related attribute but also considers the highly related background factor which affects the results indirectly. 
    \item \textit{Risk Minimzer} amplifies the utility by minimizing expected loss of the model and helps to achieve the optimal privacy-utility trade-off through regularization parameter.
\end{enumerate}

\section{Methodology}

The goal of work is to develop an unbiased secured mechanism bringing fairness which is also differentially private. To achieve, before fairness is implemented, the collected data must differentially private without hampering the number of persons corresponding to the respected protected groups, i.e the data will be divided into separate groups according to the variety of protected class. \\
Let assume we have a data set with the records of $n$ individuals and the number of dependent attributes is $m$. Among the $n$ individuals, let there be $n_i$ number of individual having $i^{th} ; i =1(1)n_p$ protected attribute. The whole mechanism is divided into seven parts, they are One-hot Encoding of collected data, Query Function, Group tagging, Correlation Matrix, Fairly Iterative Shuffling (\textit{FIS}), 2-Layers Fair Classification through\textit{ Relation Impact }and  Boosting Utility through \textit{Risk Minimizer} \& Count Report Generation.   
\subsection{One-Hot Encoding}
One hot encoding \cite{b17} \cite{b4} is a method of converting categorical variables into a structured form that could be used in ML algorithms to do better performance in prediction. Now, the correct choice of encoding has a great impression on the utility guarantee. In this architecture, the preference is one-hot encoding instead of any other process of encoding to generate a fair \& privatized report with the least noise insertion and provides maximum utility. 
\subsection{Query Function }
In the first phase of this architecture, a query function \cite{b17} is used on the dataset to get the important attributes related to a particular query as shown in table \ref{tab:qf}. Though our proposed DP mechanism works for generalized group fairness problems, we have taken \textbf{Example 2} as our case study and analyzed the case of gender unbiasedness more elaborately in upcoming sections. Also, it is to be noted, This mechanism performs for both (\textbf{Example 1 \& 2}) type of Queries; related to numerical and decision-making problems and always provide the privatized DP report against user's protected attributes.

\begin{table}[ht]
\begin{tabular}{@{}p{0.2\textwidth}p{0.77\textwidth}@{}}
\toprule
\textit{\textbf{Query}} &
  \textit{\textbf{Analysis}} \\ \midrule
\textbf{Example 1}: How many Hindus live in the urban area with the score in maths greater than ninety percent? &
  The attributes 'Religion', 'living area', and 'marks in maths' will be delivered as answers from the query function. This implies that only these three attributes are important for generating the final report to that particular numerical (count-related) query. This type of Question does not need any fair decision but should be privatized against the protected attribute "Religion". \\
\textbf{Example 2}: Which students are eligible the admission based on the marks of maths? &
  Now for this type of question, not only the attribute 'marks of maths' is important, but also the gender of the students matter as we want to develop a fair mechanism. For this, it is necessary to give special importance to the attribute 'sex' in every case which is achievable by the attribute tagging method that is discussed later. For this mechanism, the query function considers the attribute 'sex' as the important attribute and returns it no matter what the question is just to maintain a fair balance between male and female candidates. \\ \bottomrule
\end{tabular}
\caption{Query Function: Assume the database is structured with records of individuals where \textit{name, age, sex,  marks of maths, marks of physics, marks of computer science, annual family income,  living area, religion} are ten attributes.}
\label{tab:qf}
\end{table}

\subsection{Correlation Matrix}

After the query function\cite{b17}, the mechanism calculates a correlation matrix that measures the relation distance of each attribute with the required events. The previously given \textbf{Examples 2} concentrate on the decision that should fair related to gender in the admission system based on the maths score. But in reality, other sensitive background factors are indirectly related to the examination results that are obtained by the students. For example, the economical status, electricity facility, tuition opportunity, etc. It is obvious that the students came from the strong economical background or high-class society usually take more than one tuition or extra classes to improve scores. On the other hand, the poor students can't afford such facilities. Not only that, in some cases, they don't have the proper environment to study like electricity, proper living area, etc. Now, these factors have a great effect on their results and should not be avoided. So, here, the correlation matrix is used to get the highly related background factors that have a great impact on the required event.
\par
For the general case in, having a dataset with $n$ rows and $k$ attributes, the applied query function returns $c$ the number of attributes which are required to generate the answer along with protected attribute (such as skin color, religion, gender, etc) and the correlation matrix returns $d$ number of related background factors that affects the results. 
If $m= c+d$ (excluding protected attribute) number of total attributes are important for generating the final report, then these $m$ attributes will tie up together to represent a single attribute and the reduced number of attributes will be $g=k-m$. Let assume, there are $2S (S>1)$ number of shufflers and if $g$ is divisible by $S$, these attributes are divided into $S$ group with $g/S$ attributes in each group. Now, If $g$ is not divisible by $S$, then extra elements will choose a group randomly without replacement. 
After that, the group tagging method is applied which helps to preserve the individuals' protected attributes after fairly iterative shuffling \cite{b17} method. As the tied-up attributes include the highly related background attribute to the given event, the following group tagging method assists the fair classification in the future.
\subsection{Group Tagging Method}
Group tagging is a method to classify the individuals based on their protected attribute where each tagged group contains $S$ group of attributes each. This process can be happened by tagging each row with their corresponding protected class. If we consider the given \textbf{Example 2} which concentrates on fairness between gender groups (we have taken male and female groups for case study) for admission, the method will work as follow.  If the first row contains a record of a female candidate, then this row is tagged by its gender group 'female'. For the male candidates, the rows are tagged by 'male'. Now rows having the same tagging element will shuffle among each other i.e In the case of female candidates, the iterative shuffling will occur only with the rows tagged as 'female' and on the other hand, rows tagged with 'male' goes for iterative shuffling separately. In this way, the records of male and female candidates never exchange with each other and help to maintain their gender in the generated report. This process never hampers the accuracy of the classifiers for taking a fair decision based on their gender group in the given problem. Also, the iterative shuffle technique establishes a strong privacy guarantee which shuffles all the records of female and male candidates separately among their tagged group. In this way, the probability of belonging the record of a particular individual to their unique ID becomes very low, and the unbiased randomization shuffles technique keeps the data secured. This technique prevents all types of linkage attacks, similarity attacks, background knowledge attacks and also reduces the exposure of user data attributes.
\subsection{Fairly Iterative Shuffling (FIS)}
A randomized mechanism is applied with Fairly Iterative Shuffling which occurs repeatedly to the given data for producing the randomized unbiased report to a particular query.
Let's have a dataset with $n$ rows containing the records of $n$ individuals among which $n_1$ persons are female and $n_2 = n-n_1$ persons are male. That means after the group tagging method, there are exactly $n_1$ rows tagged as 'female' whereas $n_2$ rows are tagged as 'male'. In this stage of our proposed architecture, each row of 'male' and 'female' groups is shuffled  iteratively by the $2S, (S>1)$ number of shufflers. Before shuffling each gender group is divided into some batches containing $S$ attribute groups and every attribute group of each batch chooses a shuffler from the $S$ number of shufflers randomly without replacement for independent shuffling. Now the female tagged group contain $n_1$ individuals and assume that it is divided into $t_F$ batches where $1, 2,...., t_F$ batches have $n_{11}, n_{12}, ...., n_{1t_F}; (n_{11}\simeq n_{12}\simeq ....\simeq n_{1t_F})$ number of rows respectively. Similarly, let assume the male tagged group have $n_2$ individuals' records and it is divided into $1, 2, ...., t_M$ batches having $n_{21}, n_{22}, ...., n_{2t_M}; (n_{21}\simeq n_{22}\simeq ....\simeq n_{2t_M})$ rows respectively. Now for one gender tagged group, each attribute group of every batch chooses a shuffler randomly without replacement from any $S$ number of shufflers. On the other hand, For another gender-tagged group, each group of attributes chooses a shuffler similarly from the remaining $S$ number of shufflers. This shuffling technique occurs repeatedly or iteratively until the last batches of both the gender tagged group goes for independent shuffling. \\

This architecture is applicable in the generalized case also and not only limited to gender unbiasedness problem If there is $i$ number of variety for a protected attribute (skin color, religion, etc), the individuals tagged with their considerable protected class and then divided into $i$ number of groups. Each group divided into batches ($t_i$ number of batches for $ith$ group) with approximately same batch size ($n_i1 \simeq n_i2 \simeq ...\simeq n_{it_i}$) and go for shuffling \cite{b35} \cite{b17}.
\subsection{2-Layers Fair Classification through Relation Impact.}
This mechanism uses $2$ layer classifications. After $FIS$, firstly the classifiers use the group classifications that classify the individuals into gender groups (male and female) from the shuffled data sets. After that, the classifiers make the subgroup-classifications between the groups depends on the Relation Impact (details in section \ref{Fair}).
\subsubsection{Group Classification: }
After $FIS$ the classifiers divided the shuffled data into the number of classes according to their gender. In our problem setting, the mechanism focuses on the gender unbiasedness between males and females. This step is the 1st layer of classification where the classifiers use the tagged gender label to classify the data. As in the group tagging method each individual is tagged with his gender category, the classifiers concentrate on this tagged label for classification. 
This classification is made based on user-given gender information which cannot be affected or changed by the shuffling as it is tied up before going for $IS$. So, the chance of miss classification error is very less (in fact negligible) here.

\subsubsection{Subgroup-Classification}
This is the second layer of classification where the classifier classifies the groups into positive and negative subgroups by the notion of Relation Impact. The classifiers decide the eligibility, i.e if an individual should be allowed to be admitted to the course based on their related background information. For example $2$, the classifiers make the decision based on the student's math score, their economical background, living area (rural/ urban), electricity facility, etc. The positive class of each gender group contains the eligible individuals and the negative class contains the rejected individuals for the admission. 
 \subsection{Risk Minimization and Count Report Generation }
This is an important stage of this architecture where the report is generated from a shuffled and unbiased dataset for a given q
query related to sum or count aggregation i.e.numerical answer (consider the \textbf{Example 1}). The report is generated by average count calculation by an aggregator function that uses TF-IDF\cite{b15} \cite{b17} calculation to add minimum noise to the report. But, before the final report generation, the loss function \cite{b17} is calculated, and risk minimizing assessment \cite{b17} is done to obtain the optimal solution. This optimal solution refers to the situation where the architecture achieves a strong privacy guarantee, as well as maximum utility\cite{b27} measure while providing an unbiased report against a protected class. The details of risk minimization technique is discussed in section \ref{risk}.
 
\section{Analysis}
Considering a data set containing $n$ rows and $g$ attributes and there exists a $iS$ number of shufflers. The attributes are divided into $S$ groups as described previously. Now there is $n_i$ number of individuals tagged as $i^{it}$ protected attribute; $i = 1,2, .. n_p$ where each row of the dataset has the records of individuals corresponding to their crowd ID. The tagged attributes divide the data sets into groups, For example, if there is $i$ number of variety in a protected attribute, tagged individuals will be divided into $i$ groups. Now, Each tagged group are divided into some batches containing approximately the same rows i.e. $1, 2, ..., t_i$ number of batches for group tagged with $i^{th}$;$i=1(1)n_p$  attribute. After that, each group of attribute from every batch chooses a random shuffler and go for shuffling.
\subsection{FIS Randomised Response Ratio and Privacy Budget}

To proceed further with the proof, we will consider the following theorem from BUDS\cite{b17}:\\

\begin{theorem}\label{thm 1}{(Iterative Shuffling : IS \cite{b17})}
A randomisation function $\mathcal{R^*_S}$ applied by $S\ (S > 1)$ number of shuffler providing iterative shuffling to a data set $X$ with $n$ rows and $g$ attributes, where the data base is divided into $1,2, ..., t$ batches containing $n_1,n_2, ..., n_t$ number of rows respectively, will provide $\epsilon$-differential privacy to the data with privacy budget-
\begin{equation}
    \epsilon = \ln{(RR_\infty)} = \ln{\bigg[\frac{t}{(n_1 -1)^S}\bigg]}
\end{equation}
only when, $n_1 \simeq n_2 \simeq ..... \simeq n_t$.\\ Here $RR_{\infty}$= Randomised Response Ratio, and 
 \begin{equation}
    RR = \frac{P(Response  =  YES \ | \ Truth = YES)}{P(Response = YES \ | \ Truth = NO)}
 \end{equation}
\end{theorem}

\newtheorem{Lemma}{Lemma}
 \begin{Lemma}

 \textit{The Randomised Response Ration of Iterative  Shuffling($RR_{IS}$) \cite{b4} \cite{b5}=}  $\frac{P(Row \ belongs \ to \ its \ own \ unique \ ID)}{P(Row \ does \ not \ belongs \ to \ its \ own \ unique \ ID)} $ \textit{satisfy the condition of Differential Privacy.}
 \end{Lemma}

Now, the proposed architecture provides a fair mechanism along with a strong privacy budget to secure the sensitive information of individuals. This algorithm is the non-discriminant of biases regarding different groups based on protected attributes (For example gender, skin color, etc). Now the help of theorem 1, the privacy budget of \textit{FIS} is developed next. 


\begin{theorem}\label{thm 2}\textit{(Fairly Iterative Shuffling- FIS)} 
A randomisation function $\mathcal{R^*}_{iS}$ applied by $iS\ (S > 1)$ number of shuffler providing iterative shuffling to a data set $X$ with $n$ rows and $g$ attributes among which $n_i$ number of individuals are tagged with $i$th protected class ($i=1,2,...n_p$); i.e. there are $n_p$ number of user groups according to their protected class where tagged  groups are divided into $1,2, ..., t_i$  batches containing $n_{i1},n_{i2}, ..., n_{it_i}$ number of rows for $i$th group, will provide $\epsilon$-differential privacy to the data with privacy budget-
\begin{equation}
    \epsilon_{Fair} =\ln\bigg[\frac{1}{n_p !}{\Sigma_{i=1}^{n_p}\frac{t_i}{(n_{it_i} -1)^S}}\bigg]
\end{equation}

only when, $n_{i1} \simeq n_{i2} \simeq ..... \simeq n_{it_i}$; $i = 1,2,..., n_p$
\end{theorem}
\begin{proof}
The records in the group of attributes for the different tagged groups go for independent shuffling separately. According to the theorem \ref{thm 1} \cite{b17} we can say that, For the records tagged with $i$th group, 
\begin{align}
    RR_i &= \frac{P(Row \ belongs \ to \ its \ own \ unique \ ID\ tagged\ with \ ith \ group)}{P(Row \ does \ not \ belongs \ to \ its \ own \ unique \ ID\  tagged\ with \ ith\ group)}\\& =\frac{t_i}{(n_{it_i} - 1)^S}
\end{align}

as,$n_{i1} \simeq n_{i2} \simeq ..... \simeq n_{it_i}$; $i = 1,2,..., n_p$.
So, For the whole architecture the fair randomized response will be:
\begin{equation}
    RR_\infty =\frac{1}{n_p !}\bigg[\frac{t_1}{(n_{1t_1} - 1)^S} + \frac{t_2}{(n_{2t_2} - 1)^S} +...+ \frac{t_{n_p}}{(n_{n_pt_{n_p}} - 1)^S}\bigg] 
\end{equation}
   Now the privacy budget for this proposed architecture is: 
   \begin{equation}
    \epsilon_{Fair} =\ln\frac{1}{n_p !}{\bigg[\frac{t_1}{(n_{1t_1} -1)^S} + \frac{t_2}{(n_{2t_2}-1)^S} + ...+\frac{t_{n_p}}{(n_{n_pt_{n_P}} -1)^S}\bigg]}
\end{equation}
That means,
\begin{equation}
    \epsilon_{Fair} =\ln\bigg[\frac{1}{n_p !}{\Sigma_{i=1}^{n_p}\frac{t_i}{(n_{it_i} -1)^S}}\bigg]
\end{equation}

\end{proof}

In our case study, we focuses on the decision problem of gender unbiasedness between male and female groups for allowing them to take admission in a particular course (considering example 2). For this case, the $FIS$ posses a privacy budget of :
\begin{equation}
    \epsilon = \ln\frac{1}{2}{\bigg[\frac{t_F}{(n_{1t_F} -1)^S} + \frac{t_M}{(n_{2t_M}-1)^S}\bigg]}
\end{equation}{}

Where, $S>1$, $n_1$is the number of female in database, $n_2 = (n-n_1)$ is number of male in database,  $t_F$ is the number of batches in female group, $t_M$ is the number of batches in male group, $n_{1t_F}$ is the average female batch size, $n_{2t_M}$ is average male batch size and $\epsilon_{Fair}$ is the privacy parameter of this mechanism. The small value of $\epsilon_{Fair}$ refers to a better privacy guarantee.

\subsection{Fair Classifiers and Relation Impact}\label{Fair}
After $FIS$, we introduce the notion of \textit{Relation Impact} that helps the classifiers to take a fair decision. This mechanism applies fair classifiers that use the principle of Relation Impact to avoid the miss classification error and able to create unbiased classes or groups to consider for giving chance in an event.
\begin{Definition}{Relation Impact: }
The aim of a set of classifiers is to study a classification function $\hat h_c:\mathcal{X}\longrightarrow\mathcal{Y}$, defined in a hypothesis space $\mathcal{H}$ where $\hat h_c$ minimizes some aimed loss function to reduce miss-classification error based on the related background attribute class $\mathcal{A}=(a_1, a_2, ...., a_{m})$ :
\newcommand{\argmin}{\operatornamewithlimits{argmin}}
  \begin{equation}
    \hat h_c:= \argmin_{h_c\in \mathcal{H}}E_{(\mathcal{X},\mathcal{Y}) \sim P}L[{h_c(\mathcal{X} | \mathcal{A}), \mathcal{Y}}]  
  \end{equation}
The Efficiency of the classifier will be estimated using testing dataset $\mathcal{D}_{tst} = {{(x_{j_t},y_{j_t})}^{n_t}}_{j=1 (1) n_t}$,based on the observation that  how accurate the predicted labels $\hat h_c(x_{j_t}| \mathcal{A})$’s, are corresponds to the true labels $y_j$’s. 
\end{Definition}

\par Therefore, it means the decision regarding the relevant persons or groups is taken according to their due based on all the highly related factors to that particular event. This implies that distribution will not possess numerical equality whereas it will be influenced by the relation distance of important background factors and consider the individuals according to their rightful needs. For example, in the classroom, this might mean the teacher spending more time with male students at night classes rather than female students because most of the time female students do not choose to take the night classes for safety issues. Not only that, the students (irrespective of gender) who came from the weaker economical background cannot always participate in the extra classes or cannot afford any extra books or study material. So, it is clear that, in the admission system, if someone wants to choose the students fairly, not only the sex group matters, but also the economical background, and other factors like appropriate subject scores, living locations, electricity facility, etc also matters. So, these types of factors which are directly and indirectly highly related to the final exam score of the students should be taken to account. The notion of \textit{Relation Impact} exactly does that, where it trains the model by considering all the highly related attributes for taking the decision and minimizes the expected error of the classification by \textit{Equalizing Minimal FPR-FNR} of different groups. \\

In our work, the idea is to provide a fair decision with the help of an unbiased classification based on the correlation of the attributes to the required event. The mechanism first calculates the correlation matrix and any attributes which have the higher (usually greater than 0.5 or less than -0.5)  positive and negative correlation to the required events are only taken to be accountable for the final decision-making procedure. Here the classifiers use the Relation Impact where the main aim is to build a model that can minimize the expected decision loss to reach the maximum accuracy. The classifiers attend that minimum loss, i.e minimum miss classification error, by achieving the Minimal Equalised FPR-FNR of the different groups; FPR: False Positive Rate, FNR: False Negative Rate. \\


\begin{Definition}{  Equalized Minimal FPR-FNR: }
Let's have a set of classifiers $\hat{h_c} \{h_{c1}, h_{c2}, ...h_{ci}\}: \mathcal{X}\rightarrow \mathcal{Y}$ based on the hypothesis space $\mathcal{H}$ where $i= 1, 2.. n_p$ denotes the number of groups, and $\mathcal{A}=\{A_1, A_2, ... A_{m}\}$ is the set of  related background attributes, then the classifier is said to satisfy the condition of  Equalised Minimal FPR-FNR ratio if :
\newcommand{\argmin}{\operatornamewithlimits{argmin}}
\begin{equation}
    \argmin_{h_c \epsilon \mathcal{H}}FPR_{1} (\mathcal{A}) = \argmin_{h_c \epsilon \mathcal{H}} FPR_{2} (\mathcal{A}) =..........= \argmin_{h_c \epsilon \mathcal{H}}FPR_{n_p} (\mathcal{A})
\end{equation}
and, 
\begin{equation}
    \argmin_{h_c \epsilon \mathcal{H}}FNR_{1} (\mathcal{A}) = \argmin_{h_c \epsilon \mathcal{H}} FNR_{2} (\mathcal{A}) =..........= \argmin_{h_c \epsilon \mathcal{H}}FNR_{n_p} (\mathcal{A})
\end{equation}
Here,  $\hat{h_c}(\mathcal{X})$ is the predicted label by the classifier, $\mathcal{Y}$ is the true label  and both of them take the label value either $y_1$ or $Y_0$ where $y_1$ denotes the positive label, $y_0$ denotes the negative label and 
\begin{equation}
    FPR_{i} (\mathcal{A})= P_i (\hat{h_c}(\mathcal{X})= y_1 | \mathcal{Y} = y_0, \mathcal{A} )
\end{equation}
 and 
 \begin{equation}
     FNR_{i} (\mathcal{A})= P_i (\hat{h_c}(\mathcal{X})= y_0 | \mathcal{Y} = y_1, \mathcal{A} )
 \end{equation}
\end{Definition}
For our mentioned problem in example $2$, the classifier will achieve the higher accuracy for predicting the positive and negative subgroups for both the gender groups (Male and Female), by minimizing the target loss when , 
\newcommand{\argmin}{\operatornamewithlimits{argmin}}
\begin{equation}
    \argmin_{h_c \epsilon \mathcal{H}}FPR_{F} (\mathcal{A}) = \argmin_{h_c \epsilon \mathcal{H}} FPR_{M} (\mathcal{A})
\end{equation}
and 
\begin{equation}
     \argmin_{h_c \epsilon \mathcal{H}}FNR_{F} (\mathcal{A}) = \argmin_{h_c \epsilon \mathcal{H}} FNR_{M} (\mathcal{A})
\end{equation}
\newtheorem{Remark}{Remark}
  \begin{Remark}{ FIS does not hamper the Equalised Minimal FPR-FNR (Detailed discussion is in appendix)}
  
  \end{Remark}

\subsection{Boosting Utility Through Risk Minimization } \label{risk}
 The real information which can be gained from the data by a particular query or set of queries is defined as the utility \cite{b27} of the system. The goal of this section is to discuss the gained utility from the data by applying Fair-BUDS and provide a tight bound for loss function between input and output average-count. This bound has a great impact on the utility of the system and this depends on the previously obtained privacy budget. Before obtaining the final result, an optimization function with a risk assessment technique\cite{b17} is applied to get the maximum utility from the data. 
 Let's assume, before $FIS$ shuffling, the true data set will give a average count report for a time horizon $[d] = \{1,..d\}$ which is denoted as $\mathcal{C_T} = \sum_{i=1}^{n}\sum_{j\epsilon \mathcal{Q} }\sum_{T\epsilon [d]} x_{ij}[T]$; where, $x_{ij}$ is the jth record of the ith individual from the true data set for the time horizon $[d]$ and $\mathcal{Q} =\{ a_1, a_2, ...\}$ is the set of attributes given by the query function for a particular set query or set of queries. Now, $\mathcal{C_{FS}} = \sum_{i=1}^{n}\sum_{j\epsilon \mathcal{Q} }\sum_{T\epsilon[d]} x^{FS}_{ij}[d]$; where $x^{FS}_{ij}$ is the jth record of the ith individual from the $FIS$ data set for the time horizon $[d]$. The calculated loss function \cite{b17} regarding the input and output count will be:
 \begin{equation}
    \mathcal{L}(\mathcal{C_T},\mathcal{C_{FS}}) \leq \mathcal{C_{FS}} \times  \bigg|e^{\ln\bigg[\frac{1}{n_p !}{\Sigma_{i=1}^{n_p}\frac{t_i}{(n_{it_i} -1)^S}}\bigg]} - 1\bigg|
\end{equation}
When, $e^{\epsilon} \rightarrow 0$, the Utility $\mathcal{U}(\mathcal{C_T},\mathcal{C_{FS}} | \mathcal{X},\mathcal{Y}) \rightarrow 1$ as $\mathcal{L}(\mathcal{C_T},\mathcal{C_{FS}}) \rightarrow 0$. 
 Now the aim of this architecture is not only to be fair and secured but also it should pose a high utility guarantee which can be obtained by risk-minimizing assessment \cite{b17}. Proceeding with risk minimization technique, a mechanism $\mathcal{R^*}(iS)$ can be found (when there is $i$ number of tagged groups) which minimizes the risk \cite{b17}, where

\begin{equation}
    \mathcal{R^*}(iS)= \argmin_\text{$\mathcal{R^*}(iS) \epsilon \mathcal{H}$} Risk_\text{(Ep)} (\mathcal{R^*}(iS))
\end{equation}
\section{Conclusion And Future Scope}

The Proposed architecture provides group fairness to the user based on their protected attribute class while giving a strong privacy guarantee to their sensitive attribute. The Risk Minimizer amplifies the utility Guarantee to the system which makes the algorithm possess an optimal balance between privacy-utility-fairness. This balanced architecture full-filled user's top priorities related to fairness and Privacy in the Online Service system. Though the mechanism shows good performance in privacy-utility trade-off for all kinds of Queries and generates unbiased classification for decision-making problems, the use of One-Hot encoding is one kind of constraint for big data analysis. Trying other encoding option (which does not depend on data dimension) is our future target. On the other hand, the work performance and both the privacy-utility upper bounds are given on theoretical aspects only. So, various experiments with different benchmark datasets for both online and offline settings are also in our plan.


\newpage
\appendix
\section{Appendix}
 \subsection{Proof of Lemma 1 in Section 3}
 \begin{proof}
 According to the concept of Differential Privacy \cite{b5} \cite{b4}, the generated report from two neighboring datasets will be almost the same.Dwork has already shown in her privacy book that the randomized response ratio i.e; after applying the randomized function to the input dataset the ratio between changed and unchanged rows of two neighboring datasets (input and output datasets) satisfies all the notion of Differential Privacy. Dwork considered the coin toss mechanism as the randomize function which was implemented on the input dataset and the ratio between probabilities to change or unchanged of a particular response by the occurring of heads and tails, was denoted as the randomized response ratio that not only helps to derive the privacy budget but also satisfies the condition of DP, where ;
 \begin{equation}
    RR = \frac{P(Response  =  YES \ | \ Truth = YES)}{P(Response = YES \ | \ Truth = NO)}
 \end{equation}
 Now, the main idea is to prove that the intentions of  Dwork's Randomised response ratio and the IS Randomised response Ratio(RR) are almost the same. By proving this, it can be said that IS Randomised response Ratio also satisfies all the notion of DP \cite{b5} \cite{b4}. \\
 In this work, the randomized function uses the Fairly Iterative Shuffling algorithm to make secured and unbiased reports from the input dataset. For this, if we can prove that the Iterative Shuffling\cite{b17} randomized technique itself reassures the condition of DP; it will work for a Fair shuffle algorithm. For sake of simplicity, let's assume the Randomised function with IS algorithm is applied on the input data set and it only interchange two rows of the input data set. Then it is obvious that the input and output datasets now only differ from two rows and it can be said that these two datasets are neighboring. Then after applying IS, the ratio between the probabilities of change and unchanged of a single row must be denoted as the randomized response ratio which has the same intention as the coin toss randomized response ratio has. That means This is also a ratio between the probabilities of changed and unchanged reply after implementation of the randomized function. Then, surely, IS randomized Response Technique must satisfy all the notions of DP. \\
 \end{proof}
 \subsection{Remark 1 in Section 4 - Discussion}
 $FIS$ shuffles the data set in a way that does not change the position of the related background factors of the individuals which are used to decide to choose an individual in an event. Before \textit{FIS}, query function is applied and Correlation Matrix is calculated which helps to extract the related attributes to an event and tied them up together to behave like a single attribute. Now if we only consider these tied up attributes, that can be denoted as the sub-database, Let, assume the input database is $\mathcal{X}$ and $x\subset \mathcal{X}$ is the generated subset by query function and Correlation Function which contains all the related background factors/ attributes along with their tagged groups (In our case study, gender groups: male and female). Now, FIS does not change these related attributes record positions after shuffling,i.e the related attributes records belong to their user-id corresponds to their tagged group. Whereas, \textit{FIS} changes the record position of unrelated attributes only which prevents the linkage attack and minimizes the exposure of data attributes. So, after $FIS$, the individual posses all the true records of related attributes and own tagged group for consideration, whereas it posses all the false and changed records of those attributes which are not related to taking the decision. So, the sub-database $x$ does not change after $FIS$ and that helps the fair classifiers to take the decision correctly, The fair classifiers classify the data based on the sub-database $x$ and their corresponding tagged groups. As \textit{FIS} does not change $x$ and tagged group of individuals, it doesn't hamper the FPR-FNR of the different groups. 
 \subsection{Utility Guarantee and Risk Minimizer in Section 5}
 The real information which can be gained from the data by a particular query or set of queries is defined as the utility of the system. The goal of this section is to discuss the gained utility from the data by applying Fair-BUDS and provide a tight bound for loss function between input and output average-count. This bound has a great impact on the Utility of the system and this depends on the previously obtained privacy budget. Before obtaining the final result, an optimization function with a risk assessment technique is applied to get the maximum Utility from the data. All the necessary calculation regarding this is done in this section. At the end of this section, an optimum upper bound of the loss function can be achieved by which the utility will be highest while having a strong privacy budget.\\
 we have $n$ number of individual's records tagged with two separate groups based on their protected attribute class $n_i; i=i(1)n_p$ rows for $ith$ class. Now the proposed architecture generates a fairly iterative shuffling ($FIS$) method to produce the desire reports. This method, not only provides a strong privacy guarantee but also gives higher utility measures. To, prove this, the following calculations will show, how the true data average count reports are very close to the $FIS$ data average count. Let's assume, before $FIS$ shuffling, the true data set will give a average count report for a time horizon $[d] = \{1,..d\}$ which is denoted as $\mathcal{C_T} = \sum_{i=1}^{n}\sum_{j\epsilon \mathcal{Q} }\sum_{T\epsilon [d]} x_{ij}[T]$; where, $x_{ij}$ is the jth record of the ith individual from the true data set for the time horizon $[d]$ and $\mathcal{Q} =\{ a_1, a_2, ...\}$ is the set of attributes given by the query function for a particular set query or set of queries. Now, $\mathcal{C_{FS}} = \sum_{i=1}^{n}\sum_{j\epsilon \mathcal{Q} }\sum_{T\epsilon[d]} x^{FS}_{ij}[d]$; where $x^{FS}_{ij}$ is the jth record of the ith individual from the $FIS$ data set for the time horizon $[d]$. For the query in \textbf{Example 2}, we can take a attribute set $\mathcal{A}= \{a_1, ..., a_m \} $ that is highly related to the admission decisions. This set not only contains the "Maths Score" but also has other related factors like "Annual Income", "Electricity Facility" etc. All the attributes belonging to set $\mathcal{A}$ is the eligibility criteria for the students and these attributes will be tied up together to present a single attribute. Every individual from the data set will be tagged with their gender group to keep a fair balance between male and female. After $FIS$, the average count will show how many students (irrespective of gender) are eligible for admission. The fairness algorithm will be applied after that to choose the deserving students for admission based on background factors by keeping fairness among gender groups. At the time of $FIS$, the important attributes stay together always and do not shuffle among their shelves as they behave like single attributes. On the other hand, each row of a particular batch from the tagged gender groups gets shuffled and at the end records of their individuals do not belong to their unique ID. But, as the only source of noise is one-hot encoding here which is minimum, and also the tied attribute records do not break their tie at the time of shuffling it can be said that the shuffled dataset must be a neighboring dataset of the true one. Using this agenda, it also can be said that the difference between $\mathcal{C_T}$ and $\mathcal{C_{FS}}$ is negligible and these two average count results belong to the neighborhood of each other;i.e. so close to each other. By the idea of the neighboring data set, we can express the relation between two average count reports by the following bound for generalized case:
 \begin{equation}
    \mathcal{C_T} \leq e^\text{$\epsilon_{Fair}$} \mathcal{C_{FS}}
\end{equation}
where $\epsilon_{Fair}$ is the proposed privacy budget. 
when $\epsilon_{Fair}$ = 0 $\implies$ $\mathcal{C_T} = \mathcal{C_{FS}}$; i.e. The utility $\mathcal{U}(X,X^{FS})$ is maximum; where $X$ is the true data set or input dataset and $X^{FS}$ is the $FIS$ dataset or output dataset. If the range of utility is taken as $[0,1]$; when $\epsilon_{Fair}$ = 0 then  $\mathcal{U}(X,X^{FS}) = 1$. Here, 
\begin{equation}\label{a}
    \mathcal{C_T} \leq e^{\ln\bigg[\frac{1}{n_p !}{\Sigma_{i=1}^{n_p}\frac{t_i}{(n_{it_i} -1)^S}}\bigg]} \mathcal{C_{FS}}
\end{equation}
By subtracting $\mathcal{C_{FS}}$ from both sides and taking absolute value of the equation \ref{a} we get.
\begin{equation}
    |\mathcal{C_T} - \mathcal{C_{FS}}| \leq \mathcal{C_{FS}} \times  \bigg|e^{\ln\bigg[\frac{1}{n_p !}{\Sigma_{i=1}^{n_p}\frac{t_i}{(n_{it_i} -1)^S}}\bigg]} - 1\bigg|
\end{equation}
In our case study based on gender unbiasedness problem, for the Given query in \textit{Example 2} we will have :
\begin{equation}
    |\mathcal{C_T} - \mathcal{C_{FS}}| \leq \mathcal{C_{FS}} \times  \bigg|e^{\ln{\frac{1}{2}\bigg[\frac{t_F}{(n_{1t_F} - 1)^S} + \frac{t_M}{(n_{2t_M} - 1)^S}\bigg] }} - 1\bigg|
\end{equation}
Where, $S>1$, $t_F$ is the number of batches in the female group, $t_M$ is the number of batches in the male group, $n_{1t_F}$ is the average female batch size, and  $n_{2t_M}$ is average male batch size.
Now, define the loss function $\mathcal{L}(\mathcal{C_T},\mathcal{C_{FS}})$ = $|\mathcal{C_T} - \mathcal{C_{FS}}|$ and get bound on loss function as following:
\begin{equation}
    \mathcal{L}(\mathcal{C_T},\mathcal{C_{FS}}) \leq \mathcal{C_{FS}} \times  \bigg|e^{\ln{\frac{1}{2}\bigg[\frac{t_F}{(n_{1t_F} - 1)^S} + \frac{t_M}{(n_{2t_M} - 1)^S}\bigg] }} - 1\bigg|
\end{equation}
Now the aim of this architecture is not only to be fair and secured but also it should pose a high utility guarantee which can be obtained by risk-minimizing assessment. So the next target is to derive the risk function and minimize it to achieve the desired goal. \\

 Here a hypothesis class $\mathcal{H}$ must exist containing all possible Randomisation function which we are searching for. According to the concept of decision theory, the main idea is to minimizing the risk function to find the best randomization function that will help us to map best from input to output.
Here the risk function can be stated as:\\
\begin{equation}
    Risk(\mathcal{R^*}(iS)) = E[\mathcal{L}(\mathcal{C_T},\mathcal{C_{FS}})]= \int \int P(\mathcal{C_T},\mathcal{C_{FS}}) \mathcal{L}(\mathcal{C_T},\mathcal{C_{FS}}) \ d\mathcal{C_T} \ d\mathcal{C_{FS}}
\end{equation}
Where , $i$= Number of variety in protected attribute class, $P(\mathcal{C_T},\mathcal{C_{FS}})$ is the sample dataset distribution which contain $n$ data points that are drawn randomly from a population that follows the distribution $\mu(\mathcal{Z})$ over $\mathcal{Z} = \mathcal{C_T} . \mathcal{C_{FS}} : (\mathcal{C}_{\mathcal{T}1}  , \mathcal{C}_{\mathcal{T}1} ), (\mathcal{C}_{\mathcal{T}2}  , \mathcal{C}_{\mathcal{FT}2}), ...., (\mathcal{C}_{\mathcal{T}n} , \mathcal{C}_{\mathcal{FS}n})$ and 
\begin{equation}
  P(\mathcal{C_T},\mathcal{C_{FS}}) = P(\mathcal{C_{FS}}|\mathcal{C_T}) . P(\mathcal{C_T})
\end{equation}
Here the empirical risk function is: 
\begin{equation}
    Risk_\text{(Ep)} (\mathcal{R^*}(iS)) = \frac{1}{n} \Sigma_{1}^{n} \mathcal{L}(\mathcal{C}_{\mathcal{T}i},\mathcal{C}_{\mathcal{FS}i}) \leq  \frac{1}{n} \Sigma_{1}^{n} e^\text{$\epsilon_{Fair}$}  \mathcal{C}_{\mathcal{FS}i}
\end{equation}
Now a regularization parameter $\mathbf{G}$ will be added in order to impose the complexity penalty on the loss function and prevent over fitting in the following way:
\begin{equation}
     Risk_\text{(Ep)} (\mathcal{R^*}(iS)) =\frac{1}{n} \Sigma_{1}^{n} \mathcal{L}(\mathcal{C}_{\mathcal{T}i},\mathcal{C}_{\mathcal{FS}i}) + \lambda \mathbf{G}(\mathcal{R^*}(iS))
\end{equation}
\begin{equation}
      Risk_\text{(Ep)} (\mathcal{R^*}(iS)) \leq  \frac{1}{n} \Sigma_{1}^{n} e^\text{$\epsilon_{Fair}$} \mathcal{C}_{\mathcal{FS}i} + {\lambda} \mathbf{G}(\mathcal{R^*}(iS))
\end{equation}
where ${\lambda}$ manages the strength of complexity penalty. Now the mechanism $\mathcal{R^*}(iS)$ can be found which minimizes the risk, where

\begin{equation}
    \mathcal{R^*}(iS)= \argmin_\text{$\mathcal{R^*}(2S) \epsilon \mathcal{H}$} Risk_\text{(Ep)} (\mathcal{R^*}(iS))
\end{equation}
For our case study the equation will be :

\begin{equation}
    \mathcal{R^*}(iS)= \argmin_\text{$\mathcal{R^*}(2S) \epsilon \mathcal{H}$} Risk_\text{(Ep)} (\mathcal{R^*}(iS))
\end{equation}
As there are $2$ types of a protected class in related attributes (Gender): Female and Male.


\end{document}